\documentclass{article}             
\usepackage{graphicx}
\usepackage{amsmath}
\usepackage{amsthm}
\usepackage{cite}
\usepackage{accents}
\usepackage{mathtools}
\usepackage{amssymb}
\usepackage{stackengine}
\usepackage{comment}
\usepackage{geometry}
\usepackage[most]{tcolorbox}
\usepackage{xcolor}
\usepackage{todonotes}
\usepackage{arydshln}  
\newtheorem{remark}{\bf Remark}
\theoremstyle{definition}
\newtheorem{theorem}{\bf Theorem}
\newtheorem{lemma}{\bf Lemma}

\newtheorem{problem}{\bf Problem}

\theoremstyle{plain}
\newtheorem{assumption}{\bf Assumption}
\usepackage[framemethod=TikZ]{mdframed}

\newcommand{\paren}[1]{\ensuremath{\left( #1\right)}}

\newcommand{\set}[1]{\ensuremath{\left\{ #1\right\}}}

\newcommand{\norm}[1]{\ensuremath{\left\| #1\right\|}}
\newcommand{\snorm}[1]{\ensuremath{\| #1\|}}
\newcommand{\abs}[1]{\ensuremath{\left| #1\right|}}
\newcommand{\tr}{\mathrm{tr}}
\newcommand{\real}{\mathbb{R}}

\newcommand{\opnorm}[1]{\ensuremath{\| #1\|_{\mathrm{op}}}}
\newcommand{\hsnorm}[1]{\ensuremath{\| #1\|_{\mathrm{F}}}}
\newcommand{\hinfnorm}[1]{\ensuremath{\| #1\|_{\mathcal{H}_{\infty}}}}
\newcommand{\E}{\mathbb{E}}

\renewcommand{\Pr}{\mathbb{P}}

\newcommand{\du}{d_{{u}}}
\newcommand{\Du}{D_{{u}}}

\newcommand{\sigmau}{\sigma_{{u}}}
\newcommand{\sigmae}{\sigma_{{e}}}
\newcommand{\dy}{d_{{y}}}
\newcommand{\de}{d_{{e}}}
\newcommand{\Np}{N_p}
\newcommand{\Ntot}{N_\mathrm{tot}}
\newcommand{\cplx}{\mathbb{C}}

\newcommand{\diag}{\mathrm{diag}}

\newcommand{\qq}{\mathrm{\mathstrut q}}
\newcommand{\A}{\mathcal{A}}
\title{\LARGE \bf
Finite Sample Frequency Domain Identification
}
\usepackage{hyperref}
\hypersetup{colorlinks,
            linkcolor=blue,
            citecolor=green!50!black}
\usepackage{authblk}
\author[1]{Anastasios Tsiamis}
\author[1]{Mohamed Abdalmoaty}
\author[1]{Roy S. Smith}
\author[1]{John Lygeros}
\affil[1]{ETH Zürich, Switzerland}
\begin{document}

\maketitle
\thispagestyle{empty}
\pagestyle{empty}

\begin{abstract}
We study non-parametric frequency-domain system identification from a finite-sample perspective. We assume an open loop scenario where the excitation input is periodic and consider the Empirical Transfer Function Estimate (ETFE), where the goal is to estimate the frequency response at certain desired (evenly-spaced) frequencies, given input-output samples. We show that under sub-Gaussian colored noise (in time-domain) and stability assumptions, the ETFE estimates are concentrated around the true values. The error rate is of the order of $\mathcal{O}((\du+\sqrt{\du\dy})\sqrt{M/\Ntot})$, where $\Ntot$ is the total number of samples, $M$ is the number of desired frequencies, and
 $\du,\,\dy$ are the dimensions of the input and output signals respectively. 
 This rate remains valid for general irrational transfer functions and does not require a finite order state-space representation. By tuning $M$, we obtain a $\Ntot^{-1/3}$ finite-sample rate for learning the frequency response over all frequencies in the $
\mathcal{H}_{\infty}$ norm.   
Our result draws upon an extension of the Hanson-Wright inequality to semi-infinite matrices. We study the finite-sample behavior of ETFE in simulations. 

\end{abstract}
\section{Introduction}\label{sec:introduction}
We consider the identification of \emph{unknown} linear, discrete-time, time-invariant systems of the form
\begin{equation}\label{eq:system}
\begin{aligned}
     y_t&=\bar{y}_t+v_t,\\
  \bar{y}_t&= G(\qq) u_t,\,  G(\qq)\triangleq \sum_{s=0}^{\infty}g_s \qq^{-s},
\end{aligned}
\end{equation}
where $t\in\mathbb{Z}$ is the time, $u_t\in\mathbb{R}^{\du}$ is the input, $y_t\in\mathbb{R}^{\dy}$ is the output, $\qq^{-s}u_t=u_{t-s}$ is the backward shift operator, and $g_t\in\mathbb{R}^{\dy\times \du}$ is the impulse response. The noiseless output $\bar{y}_t$ is perturbed by some random noise process $v_t\in\mathbb{R}^{\dy}$.
We are interested in estimating the frequency response $G(e^{j\omega})$
from finite input-output data.

Frequency domain identification has been extensively studied~\cite{Ljung1999system,schoukens2004time,pintelon2012system}. The estimation error guarantees (on its distribution) are typically asymptotic, e.g. see Central Limit Theorem in~\cite[Ch. 16]{pintelon2012system}, and, thus, are valid when the number of samples grows to infinity. 
Here, we adopt a finite-sample point of view, motivated by advances in modern statistics~\cite{vershynin2018high,wainwright2019high} and statistical learning theory. 
Asymptotic methods are sharp asymptotically but are often heuristically applied for finite samples. Finite-sample bounds, on the other hand, are valid for any number of samples, 
but suffer from looser bounding constants. Nonetheless, they can provide a more detailed
qualitative characterization of the statistical difficulty of learning~\cite{tsiamis2023statistical}.

While finite-sample system identification has been studied before, most results are focused on time domain identification~\cite{goldenshluger1998nonparametric,faradonbeh2018finite_ID,simchowitz2018learning,oymak2021revisiting,sarkar2021finite,tsiamis2019finite,wagenmaker2020active,tu2022learning,ziemann2022learning,ziemann2022single}. Detailed related work and a tutorial on the subject can be found in~\cite{tsiamis2023statistical,ziemann2023tutorial}.
Frequency domain and time domain identification have many similarities--ignoring initial conditions, transients, or leakage effects, the two domains are equivalent from a prediction error framework perspective~\cite{schoukens2004time}. Still, working in one domain may offer some advantages over the other~\cite{schoukens2004time}. For example, the frequency domain approach allows a unified treatment of discrete and continuous time systems, simplifies the analysis of systems with delays, and offers a more explicit way of designing the input excitation.

\vspace{0.5em}
\noindent Our contributions are the following:
\vspace{0.5em}

\textbf{Finite-sample guarantees for the ETFE.} We provide finite sample guarantees for the well-established Empirical Transfer Function Estimate (ETFE)~\cite{Ljung1999system}, a non-parametric method for frequency domain identification, under open-loop periodic excitation. While the mean and variance of the ETFE have been characterized before, we provide guarantees on the distribution of the estimation error, the tail probabilities in particular. Under certain stability conditions, 
we prove that the estimation error decays with a rate of $\sqrt{M/\Ntot}$, where $\Ntot$ is the total number of samples. The parameter $M$ is the number of selected frequencies at which we estimate the frequency response; it controls the frequency resolution. The rate holds for general irrational transfer functions and does not require a finite order state-space representation, unlike prior non-asymptotic bounds~\cite{sarkar2021finite}.

\textbf{Guarantees in the $\mathcal{H}_\infty$ norm.} Based on our finite-sample bound, we tune the number of frequencies $M$ to provide guarantees for learning the frequency response across all frequencies. We provide a non-asymptotic rate of $\Ntot^{-1/3}$ in the $\mathcal{H}_\infty$ norm of the estimation error, which reflects optimal rates for non-parametric learning of Lipschitz functions~\cite{tsybakov2008introduction}.

\textbf{Extension of the Hanson-Wright inequality.} To prove our main result we have to deal with quadratic forms of a (countably) infinite number of sub-Gaussian variables. To achieve this, we extend the celebrated Hanson-Wright inequality~\cite{vershynin2018high,hanson1971bound} to semi-infinite matrices; that is, bounded operators mapping sequences to finite vector spaces.

Our paper is related to non-parametric system identification, which includes works on both time~\cite{care2023kernel} and frequency domain~\cite{devonport2023frequency}.  Using Gaussian Processes as in~\cite{devonport2023frequency}, where the unknown frequency response follows a Gaussian prior, we can also obtain finite sample guarantees. Here, we follow a different approach and we do not consider Gaussian priors. Note that in this work we focus on qualitative data-independent bounds linking sample requirements to system theoretic properties. Data-dependent bounds, which are arguably more suitable for applications, have also been studied before~\cite{campi2004non,ko2015non,baggio2022bayesian}.

\textbf{Notation. }Let $(H,\mathbb{F},\langle\cdot,\cdot\rangle_{H})$ be a Hilbert space with field $\mathbb{F}=\real$ or $\cplx$ and inner product $\langle \cdot,\cdot \rangle_H$. For any vector $x\in H$, let $\snorm{x}\triangleq \sqrt{\langle x,x \rangle_H}$ denote the inner product norm. Let $H,V$ be Hilbert spaces with $\mathbb{F}=\real$ or $\cplx$ and let $\mathcal{A}:H\rightarrow V$ be any linear map. Let $\opnorm{\mathcal{A}}\triangleq \sup_{\snorm{x}=1}\snorm{\mathcal{A}(x)}$ denote the operator norm and $\mathcal{A}^*$ denote the adjoint operator. If $\set{b_i}_{i\in\mathcal{I}}$ is an orthonormal basis, the Hilbert-Schmidt or Frobenius norm is defined as $\hsnorm{\mathcal A}^2\triangleq\sum_{i\in\mathcal{I}}\snorm{\mathcal{A}(b_i)}^2$. Let $\mathcal{L}_2(p)\triangleq \{\set{x_k}^{\infty}_{k=1}:x_k\in\real^p,\,\sum_{k=1}^{\infty}\snorm{x_k}^2<\infty \}$ be the Hilbert space of $p$-dimensional square summable sequences. A universal constant is a constant that is independent of the problem at hand, e.g., the system or the algorithm. For any integer $M$, let $[M]\triangleq{0,\dots,M-1}$. The $\mathcal{H}_{\infty}$ norm of $G(e^{j\omega}):[0,2\pi)\rightarrow \cplx^{d_1\times d_2}$ is given by $\sup_{\omega\in[0,2\pi)}\opnorm{G(e^{j\omega})}$; it is denoted by $\hinfnorm{G}$.

\section{Problem formulation}\label{sec:formulation}
Consider the input-output system~\eqref{eq:system}. We make the following assumption about the noise process $v_t$.
\begin{assumption}[Noise]\label{assum:noise}
    The noise process $v_t$ is filtered sub-Gaussian white noise, that is,
    \begin{equation}\label{eq:noise}
    v_t = H(\qq) e_t,\quad H(\qq)\triangleq \sum_{s=0}^{\infty}h_s \qq^{-s},
\end{equation}
where $h_t\in\mathbb{R}^{\dy\times\de}$ are the \emph{unknown} filter coefficients. Let $e_t\in\mathbb{R}^{\de}$ be i.i.d. zero mean, with covariance $\E{e_te_t^{\top}}=\sigma_e^2 I_{\de}$, and $K^2$-sub-Gaussian~\cite{vershynin2018high}, i.e., for any $\xi\in\real^{\de}$
\begin{equation}\label{eq:subgdef}
    \E \exp \left( \xi^\top e_t \right) \leq \exp\left( \frac{K^2\|\xi\|^2}{2} \right),
\end{equation}
for some $K>0$.
\end{assumption}
The noise process $v_t$ is colored. It is used to model measurement noise as well as any stochastic disturbances acting on the dynamical system.

We assume throughout that the input is bounded. This guarantees that any transient phenomena have a limited effect on the estimation problem.
\begin{assumption}[Input Bound]\label{assum:input_bounded}
 All inputs are bounded
    \[
    \norm{u_t}\le \Du,\,\text{ for all }t\in \mathbb{Z}
    \]
    for some $D_u>0$ independent of $t$.
\end{assumption}
We start all identification experiments at time $t=0$. Hence, the initial conditions are determined by all past signals $u_{-1},u_{-2},\dots$ and $e_{-1},e_{-2},\dots,$, which are nonzero in general, and unknown. 
Note that our formulation allows general irrational transfer functions and does not assume a state-space representation of finite dimension. 

\subsection{Empirical Transfer Function Estimate}
The goal of non-parametric frequency domain identification is to estimate the frequency response $G(e^{j\omega})$, given input-output data. We assume access to $\du$ experiments of length $N$, that is, data $(u^{(i)}_0,y^{(i)}_0,\dots,u^{(i)}_{N-1},y^{(i)}_{N-1})$, for $i=1,\dots,\du$. This brings the total number of samples to $\Ntot\triangleq\du N$. 
We assume that the trajectories are \emph{statistically independent} and we leave the single trajectory case for future work.

We are interested in the performance of the ETFE, which we review here.
Given any signal $z=\{z_t\}_{t\in[N]}$, let \[Z_k\triangleq \mathcal{F}^N_k(z)\triangleq \frac{1}{\sqrt{N}}\sum_{t=0}^{N-1}z_te^{-j\frac{2\pi k}{N}t},\,k\in[N]\] denote its $N$-point Discrete Fourier Transform (DFT), evaluated at $k$. 
Let $Y^{(i)}_k,U^{(i)}_k$ be the $N$ point DFTs of $y^{(i)}_t$ and $u^{(i)}_t$ respectively for the $i-$th experiment, $i=1,\dots,\du$. Let $Y_k\in\cplx^{\dy\times\du},U_k\in\cplx^{\du\times \du}$ denote the stacked DFTs for all experiments
\begin{equation}\label{eq:Yk_Uk_DFTs}
   \!\!\! Y_k\triangleq \begin{bmatrix}
        Y^{(1)}_k&\cdots&Y^{(\du)}_k
    \end{bmatrix},\,  \!U_k\triangleq \begin{bmatrix}
        U^{(1)}_k&\cdots&U^{(\du)}_k
    \end{bmatrix}\!.
\end{equation}
Then, an estimate of $G(e^{j\omega})$ at frequency $\omega_k = 2\pi k/N$, for $k =0, \dots, N-1$, can be obtained using the \emph{ETFE}
\begin{equation}\label{eq:ETFE}
\hat{G}_k \triangleq Y_kU^{-1}_k,
\end{equation}
provided that $U_k$ is invertible;
 the estimate is undefined if not. Since the number of frequencies $N$ scales with the number of data, it is generally impossible to estimate the responses at all frequencies consistently (without assuming structure)~\cite{Ljung1999system}. Instead, we can learn the responses at a smaller frequency set. Given a frequency-resolution parameter $M<N$, we  focus on estimating $G(e^{j\omega})$ at $\omega=2\pi \ell/M$, for $\ell\in[M]$. Based on the ETFE and under some additional Lipschitz assumptions on the frequency response, we can extend the estimation over all frequencies $\omega\in[0,2\pi)$.

\subsection{Excitation Method}
The estimation performance also depends on the excitation method. 
Since we only need to estimate the frequency responses at $2\pi \ell/M$, $\ell\in[M]$, it is sufficient to excite the system at only these frequencies~\cite{Ljung1999system}. Assuming that $M$ divides $N$, the DFT $U_k$ of the input can be non-zero at  only $2\pi k/N=2\pi \ell/M$ or $k=\ell N/M$. The latter condition is satisfied if and only if the excitation input is periodic with a period equal to $M$.
Note that we also need invertibility of $U_k$ at $k=\ell N/M$. To achieve this, we assume the following.

\begin{assumption}[Excitation]\label{assum:excitation}
     Let the input signals be periodic with period $M$ such that $u^{(i)}_{t+M}=u^{(i)}_{t}$, for $t\ge 0$ and every experiment $i=1,\dots,\du$. Assume that $M$ divides $N$ with $N_p\triangleq N/M$. Consider \emph{one period} of the input signals and let the respective $M$-point DFTs be
    \[
\tilde{U}^{(i)}_\ell=\mathcal{F}^M_\ell(u^{(i)})=\frac{1}{\sqrt{M}}\sum_{t=0}^{M-1}u^{(i)}_te^{-\frac{2\pi \ell}{M}t}, i=1,\dots,\du
    \]
    for $\ell\in[M]$, with respective stacked DFTs
    \[
\tilde U_\ell\triangleq \begin{bmatrix}
         \tilde U^{(1)}_\ell&\cdots& \tilde U^{(\du)}_\ell
        \end{bmatrix}.
    \]
    Assume that for all $\ell\in[M]$ the stacked DFTs satisfy   \begin{equation}\label{eq:excitation_condition}
        \sigma_{u,\ell}^2 I_{\du}\preceq \tilde{U}_\ell\tilde{U}^*_\ell,
    \end{equation}
    for some $0<\sigma^2_{u,\ell}$. 
\end{assumption}
By definition, for Assumptions~\ref{assum:input_bounded},~\ref{assum:excitation} to be consistent, we need $\sum_{\ell=0}^{M-1}\sigma^2_{u,\ell}\le M\Du^2$, where $\Du$ is the input upper bound of Assumption~\ref{assum:input_bounded}.

Such assumptions are standard when dealing with experiment design in frequency domain. For example, Assumption~\ref{assum:excitation} is satisfied by design (with uniform $\sigma^2_{u,\ell}$ across $\ell\in[M]-\{0\}$) when pseudorandom binary sequence (PRBS) signals are used and we excite one input at a time~\cite[Ch. 13]{Ljung1999system}. Another choice could be multisine signals~\cite{dobrowiecki2006optimized}, where the user simply designs the input to have sinusoids with non-zero amplitudes at the required frequencies. 
Another option is to design the input spectrum and generate the input by passing a white noise realization through the spectral factor~\cite{bombois2011optimal}.

 \subsection{Objective}
We can now state our objective, which is providing finite-sample guarantees for estimating the frequency responses. We focus on $\epsilon-\delta$ probabilistic guarantees, where $\epsilon$ controls the estimation accuracy and $\delta$ controls the confidence. 
\begin{mdframed}[roundcorner=3pt, backgroundcolor=blue!6,innertopmargin=-2pt]
\begin{problem}[Finite-Sample ETFE]\label{prob:ETFE}
Fix a frequency resolution $M<N$ such that $M$ divides $N$ and denote their ratio by  $\Np=N/M$. Consider $\du$ independent input-ouput trajectories of length $N$ $\{u^{(i)}_t,y^{(i)}_t\}_{t\in[N]}$, for $i=1,\dots,\du$,  generated by system~\eqref{eq:system} with excitation inputs as in Assumption~\ref{assum:excitation}. 
Fix a failure probability $0<\delta<1$. Determine $\epsilon_\ell>0$, $\ell\in[M]$ such that
\[
\Pr\paren{\opnorm{G(e^{j\frac{2\pi\ell}{M}})-\hat{G}_{\ell \Np}}\ge \epsilon_\ell,\,\forall \ell\in[M]}\le \delta,
\]
where the ETFE $\hat{G}_{\ell \Np}$ is defined in~\eqref{eq:ETFE}.
\end{problem}
\end{mdframed}
Problem~\ref{prob:ETFE} only focuses on the desired discretized frequency grid $\{2\pi \ell/M\}_{\ell\in[M]}$. In Section~\ref{sec:uniform_guarantees}, we also study uniform guarantees over all frequencies in the $\mathcal{H}_{\infty}$ norm.

To guarantee a well-defined estimation problem, we consider the following stability conditions.

\begin{assumption}[Strict Stability]\label{assum:strict_stability}
    The input-output impulse response is strictly stable~\cite{Ljung1999system}, that is, 
    \begin{equation}\label{eq:system_strictly_stable}
\snorm{G}_{\star}\triangleq \sum_{t=0}^\infty t\opnorm{g_t}<\infty.
    \end{equation}
    The auto-correlation function of the noise
    $R_t\triangleq \E v_{s}v_{s-t}^\top$ is also strictly stable    \begin{equation}\label{eq:autocorr_noise_strictly_stable}
\snorm{R}_{\star}\triangleq \sum_{t=0}^\infty t\opnorm{R_t}<\infty.
    \end{equation}
\end{assumption}
Strict stability guarantees that the derivative of the frequency response $\partial G(e^{j\omega})/\partial\omega$ is uniformly bounded over all frequencies. This, in turn, implies that the response $G(e^{j\omega})$ is Lipschitz. Strict stability also guarantees that the transient phenomena have a limited effect on the estimation procedure. 
\section{Finite-sample guarantees for the ETFE}\label{sec:finite_sample}
In this section, we focus on estimating the frequency response at $\{2\pi \ell/M\}_{\ell\in[M]}$, that is, the selected frequencies.
Following the convention of~\eqref{eq:Yk_Uk_DFTs}, we define the stacked DFTs of the noises and the noiseless outputs as
\[
V_k \triangleq \begin{bmatrix}
        V^{(1)}_k&\cdots&V^{(\du)}_k
    \end{bmatrix},\,\bar Y_k \triangleq \begin{bmatrix}
        \bar Y^{(1)}_k&\cdots&\bar Y^{(\du)}_k
    \end{bmatrix}
\]
Then, for every frequency $\omega_k=2\pi k/N$ we have
\begin{equation}\label{eq:system_frequency_domain}
Y_k=\bar{Y}_k+V_k=G(e^{j\omega_k})U_k+V_k+T_{k,N},
\end{equation}
where $T_{k,N}=\bar{Y}_k-G(e^{j\omega_k})U_k$ accounts for transient and time-aliasing phenomena since the DFT of $\{\bar{y}^{(i)}_t\}_{t=0}^{N-1}$ is different from $\{G(e^{j\omega_k})U^{(i)}_k\}_{k=0}^{N-1}$ for finite $N$. This term vanishes as $N$ grows to infinity.

\begin{remark}
The above relation fits the framework of non-parametric function estimation. However, there are some notable differences with standard formulations~\cite{tsybakov2008introduction,wainwright2019high}. First, we have the presence of the input $U_k$, which affects the signal-to-noise ratio (SNR) and is an additional degree of freedom. For example, if the input matrix is not invertible at some $k$, we do not get a well-defined sample of $G(e^{j\omega_k})$. Second, the noise $V_k$ is heteroscedastic since its variance depends on the frequency $k$. Moreover, the sequence $V_k,\,k\in[N]$ is non-Gaussian and non-independent across frequencies for finite samples $N$ (only asymptotically as $N$ goes to infinity). Hence, the non-asymptotic techniques of~\cite[Ch. 13]{wainwright2019high} do not apply directly. 
\end{remark}

The estimation error  is equal to
\begin{equation}\label{eq:estimation_error}
\begin{aligned}
\hat{G}_k-G(e^{j\omega_k})&=V_kU^{-1}_k+T_{k,N}U^{-1}_k,
\end{aligned}
\end{equation}
where the input matrix $U_k$ is invertible, and we only look at the frequencies $k=\ell N_p$, $\ell\in[M]$. Let $\Phi_{v,N}(k)\triangleq \E V^{(i)}_k(V^{(i)}_k)^*$ be the aliased power spectrum of the process $v_t$ at frequency $k$, where due to independence, the experiment index $i$ does not affect the definition. Define the signal-to-noise ratio (SNR) at frequency $k=\ell N_p,$ $\ell\in[M]$ as
\begin{equation}\label{eq:snr}
\mathrm{SNR}_{k,N}\triangleq \frac{\sigma_{u,\ell}}{\sqrt{\opnorm{\Phi_{v,N}(k)}}},
\end{equation}
where $\opnorm{\Phi_{v,N}(k)}$ is interpreted as the matrix norm for fixed $k$.
We obtain the following finite-sample guarantees.
\begin{theorem}[ETFE Finite-Sample] \label{thm:finite_sample_ETFE}
Consider Problem~\ref{prob:ETFE} and fix a failure probability $\delta>0$. Under Assumptions~\ref{assum:noise}-\ref{assum:strict_stability}, with probability at least $1-\delta$ for all $k=\ell N_p$, $\ell\in[M]$
  \begin{align}\label{eq:HW_ETFE}
&\opnorm{G(e^{j\omega_k})-\hat{G}_k}\le \frac{2\snorm{G}_\star D_u\sqrt{M}}{\sigma_{u,\ell} N}\\&+\frac{\sqrt{M}}{\sqrt{N}}\mathrm{SNR}^{-1}_{k,N}\paren{\sqrt{\dy}+c\frac{K^2}{\sigmae^2}\sqrt{\du+\log M/\delta}}
\nonumber
\end{align}
where $c$ is a universal constant, $\snorm{G}_\star$ is defined in~\eqref{eq:system_strictly_stable}, and $\Du$ is the maximum input norm.
\end{theorem}
The exact value of the universal constant can be found in the proof. The first term  of the right-hand side captures the transient error $T_{k,N}U^{-1}_k$, while the second one captures the error $V_kU^{-1}_k$ due to stochastic noise. Recall that the total number of samples is equal to $\Ntot=\du N$. As we increase the number of samples $\Ntot$ while keeping $M$ constant, the former term decays at a faster rate of $1/\Ntot$ compared to the latter's $1/\sqrt{\Ntot}$.  Hence, the non-asymptotic rate is
\[
\mathcal{O}\left(\frac{\sqrt{M}}{\sqrt{\Ntot}}\sqrt{\du}\big(\sqrt{\du+\log M/\delta}+\sqrt{\dy}\big)\right).
\]
The rate is similar to the ones for non-asymptotic parametric identification in time-domain~\cite{ziemann2023tutorial}; the optimal rate in that line of work is typically of the order $\sqrt{d}$ for some $d$ scaling with the number of unknown parameters. Here, we have a similar scaling of $\sqrt{M}(\sqrt{\du}+\sqrt{\dy})$ (ignoring log terms) times an additional $\sqrt{\du}$ dimensional dependence. This is an artifact of imposing a strict input norm bound in Assumption~\ref{assum:input_bounded}. If we allow $\sigmau, \Du$ to scale with $\sqrt{\du}$ (as is the case for white-noise inputs in the time-domain~\cite{ziemann2023tutorial}), we can remove this extra term. 

A benefit of frequency-domain identification is that it provides specialized guarantees for every frequency of interest by breaking down the SNR into SNRs for every frequency. This offers direct insights on which frequencies to focus on and how to design the excitation inputs. Note that the inverse $\mathrm{SNR}^{-1}_{k,N}$ is upper bounded and converges to a limit as $N$ grows to infinity. This is a consequence of the following result which exploits the strict stability condition~\eqref{eq:autocorr_noise_strictly_stable}.
\begin{lemma}[Stochastic Transient~{\cite[Sec. 6.3]{Ljung1999system}}]\label{lem:stochastic_transient}
Denote the power spectrum of the noise at frequency $\omega_k=2\pi k/N$, for some $k\in[N]$, by $\Phi_v(k)\triangleq \sum_{t=-\infty}^{\infty}R_{t}e^{-j\omega_k t}$, with $R_{-t}=R_{t}^\top$.
  We have
    \[
\sup_{k\in[N]}\opnorm{\Phi_v(k)-\Phi_{v,N}(k)}\le 2\snorm{R}_{\star}/N.
    \]
\end{lemma}
In the remainder of the section, we provide a sketch of the proof of Theorem~\ref{thm:finite_sample_ETFE}. We simply bound every term that appears in~\eqref{eq:estimation_error} separately.
\subsection{Deterministic transient}
Even in the absence of any stochastic noise, the ETFE suffers from estimation errors due to transient phenomena (e.g. aliasing, leakage)~\cite{schoukens2004time}. Fortunately,
the transient error $T_{k,N}=\bar{Y}_k-G(e^{j\omega_k})U_k$ decays uniformly to zero as the DFT horizon $N$ goes to infinity. 
\begin{lemma}[Deterministic Transient~{\cite[Sec. 2.2]{Ljung1999system}}]\label{lem:deterministic_transient}
     We have
    \[
\sup_{k\in[N]}\opnorm{T_{k,N}}\le 2\snorm{G}_{\star}\Du/\sqrt{N}.
    \]
\end{lemma}
The above result is a consequence of bounded inputs and strict stability (Assumptions~\ref{assum:input_bounded},~\ref{assum:strict_stability}). Together they guarantee that the deterministic transient phenomena have a vanishing effect on the estimation procedure.

\subsection{Input energy}
Next, we review a standard result for periodic inputs. The input  $U_k$ at frequencies $k=\ell N_p$ is equal to $\sqrt{N_p}\tilde{U}_\ell$, where $\tilde{U}_\ell$ is the $M-$point DFT based on one period of the input signals (see Assumption~\ref{assum:excitation}). As a result, $U^{-1}_k$ is well-defined and vanishes to zero with $N^{-1/2}_p$. This phenomenon is a direct consequence of periodicity and the properties of DFT.
\begin{lemma}[Input energy]\label{lem:PE}
    Let Assumption~\ref{assum:excitation} be in effect. Then, for all $k=N_p\ell$,  $\ell\in[M]$
    \begin{equation}
        \sigma_{u,\ell}^2 N_p I_{\du}\preceq U_kU^*_k,\text{ for }k=\ell N_p.
    \end{equation}
\end{lemma}
Lemma~\ref{lem:PE} is key to achieving consistency. While the known input is periodic, the noise is not. Hence, at the selected frequencies, the noise is averaged out.

\subsection{Noise concentration}
Finally, we bound the noise term $V_k$  by showing that its norm concentrates around $\sqrt{\tr(\Phi_{v,N}(k))}$. 
\begin{lemma}[Concentration of noise]\label{lem:noise_concentration} Fix a $k\in[N]$. Under Assumption~\ref{assum:noise}, for any $s>0$
    \begin{multline}
    \Pr\set{\abs{\opnorm{V_k}-\sqrt{\tr(\Phi_{v,N}(k))}}>2s\sqrt{\tr(\Phi_{v,N}(k))}}\le 9^{2\du} 2    \exp\paren{-\frac{1}{144}s^2\frac{\sigmae^4 }{K^4}\frac{\tr(\Phi_{v,N}(k))}{\opnorm{\Phi_{v,N}(k)}}}.
\end{multline}
\end{lemma}
Prior results show that $V^{(i)}_k$ is asymptotically circular Gaussian as $N$ goes to infinity~\cite[Ch. 16]{pintelon2012system}. We show that even under finite $N$, the tails of the distribution decay exponentially reflecting the properties of the Gaussian distribution.

 The proof of Lemma~\ref{lem:noise_concentration} is based on a novel extension of the Hanson-Wright inequality~\cite{hanson1971bound}, a standard tool for proving concentration of quadratic forms involving random variables.  Note that $V_k$ is a linear function of an infinite number of random noises $e^{(i)}_t$, $i=1,\dots,\du$, $-\infty<t\le N-1$. Hence, we need to extend the Hanson-Wright inequality~\cite{hanson1971bound,rudelson2013hanson} to semi-infinite matrices, that is, bounded operators from square summable (real-valued) sequences $\mathcal{L}_2$ to finite vector spaces.

\begin{theorem}[Semi-infinite Hanson-Wright]\label{thm:HW}
    Consider a sequence $z\triangleq\{z_t\}_{t=1}^{\infty},\,z_t\in\real^p$ of independent, zero-mean, $K^2$-sub-Gaussian random variables with covariance $\E z_tz^\top_t=\sigmae^2 I_p$. Let $\mathcal{A}:\mathcal{L}_2(p)\rightarrow \real^d$ be a linear map with bounded Frobenius norm $\hsnorm{\mathcal{A}}<\infty$. For any $\alpha>0$
\begin{multline}\label{eq:HW_inequality}
    \Pr\set{\abs{\snorm{\mathcal{A}(z)}^2-\E\snorm{\mathcal{A}(z)}^2}>\alpha\hsnorm{\mathcal{A}}^2}\le \\2 \exp\paren{-\min\Bigg\{\frac{\alpha^2\hsnorm{\mathcal{A}}^2}{144 K^4 \opnorm{\mathcal{A}}^2},\frac{\alpha\hsnorm{\mathcal{A}}^2}{16\sqrt{2} K^2 \opnorm{\mathcal{A}}^2} \Bigg\}},
\end{multline}
where the extension $\A(z)$ to $z$ is defined almost surely and $\snorm{\cdot}$ denotes the inner product norm.
\end{theorem}
The constants that appear in the statement are similar to~\cite{ziemann2023tutorial}; we extend their proof to semi-infinite matrices.
The Hanson-Wright inequality can be used to prove concentration of quadratic forms of the form $\xi^*V^*_kV_k\xi$. This, in turn, implies concentration of the norm $\opnorm{V_k}$ via the variational representation of the norm $\opnorm{V_k}^2=\sup_{\snorm{\xi}=1} \xi^*V^*_kV_k\xi$.

\section{Estimation over all frequencies}\label{sec:uniform_guarantees}
In the previous section, we derived finite-sample guarantees for estimating the frequency responses at fixed selected frequencies $2\pi\ell/M$,  $\ell\in[M]$. Here, we derive guarantees for estimating the function $G(e^{j\omega})$ uniformly over all $\omega\in [0,2\pi)$ in the $\mathcal{H}_{\infty}$ norm. We consider a naive estimator where to compute $\hat{G}(e^{j\omega})$ we use the closest frequency $\hat{G}(e^{j2\pi \ell/M})$, for some $\ell\in[M]$. 

Let $\chi(t)$ be the indicator function of the half-open interval $[-1/2,1/2)$.  Define the naive estimator
\begin{equation}\label{eq:naive_estimator}
\hat{G}^N(e^{j\omega})\triangleq\sum_{\ell=0}^{M-1}\hat{G}_{\ell N_p}\chi\left(\frac{\omega-\omega_{\ell N_p}}{2\pi/M}\right).
\end{equation}
Due to strict stability, it follows that the frequency response is smooth with Lipschitz constant upper bounded by $\snorm{G}_{\star}$. This follows from the fact that
\[
\norm{\frac{\partial G(e^{j\omega})}{\partial \omega}}_{\mathrm{op}}=\norm{\sum_{t=0}^{\infty} jt g_t e^{-j\omega t}}_{\mathrm{op}}\le \snorm{G}_\star.
\]
Hence, for any $\omega\in[0,2\pi)$, the error is bounded by
\[
\opnorm{G(e^{j\omega})-\hat{G}^N(e^{j\omega})}\le \frac{\pi}{M}\snorm{G}_{\star}+\opnorm{G(e^{j\omega_{\ell N_p}})-\hat{G}_{\ell N_p}},
\]
for some $\ell\in[M]$.
The first term scales with $1/M$, while the latter scales with $\sqrt{M/N}$, excluding logarithmic terms. Balancing the two terms, we obtain the following guarantees. 

\begin{theorem}[Guarantees in $\mathcal{H}_{\infty}$ norm]\label{thm:sup_norm_guarantees}
    Let $M=c_1N^{1/3}$, for some $c_1>0$ such that $M$, $N/M$ are integers. Consider the naive estimator~\eqref{eq:naive_estimator} and fix a failure probability $\delta$. Under Assumptions~\ref{assum:noise},~\ref{assum:input_bounded},~\ref{assum:excitation},~\ref{assum:strict_stability}, with probability at least $1-\delta$:
     \begin{align}
&\hinfnorm{G-\hat{G}^N}
\le \pi c^{-1}_1N^{-1/3}\snorm{G}_{\star}+c^{1/2}_1N^{-5/6}\frac{2\snorm{G}_\star D_u}{\underline{\sigma}_{u}}\nonumber\\& +c^{1/2}_1cN^{-1/3}\underline{\mathrm{SNR}}^{-1}_{N}\paren{\sqrt{\dy}+\frac{K^2}{\sigmae^2}\sqrt{\du+\log N/\delta}}\label{eq:sup_norm_guarantees}
\end{align}
where $c$ is a universal constant, \[\underline{\mathrm{SNR}}_N=\min_{\ell\in[M]}\mathrm{SNR}_{\ell N_p,N}\]
is the worst case SNR, and $\underline{\sigma}_{u}=\min_{\ell\in[M]}\sigma_{u,\ell}$ is the worst case excitation among the frequencies of interest.
\end{theorem}
We can tune the constant $c_1$ to guarantee $\Np$ is an integer and trade between the Lipschitz constant and the SNR.
Excluding logarithmic factors, we obtain a rate of $\Ntot^{-1/3}$ which is the optimal one for non-parametric estimation of Lipschitz functions~\cite{tsybakov2008introduction,wainwright2019high}. The rate is suboptimal when higher order derivatives exist, which would imply stricter stability conditions than Assumption~\ref{assum:strict_stability}. For example, for twice differentiable functions, the rate can be improved to $\Ntot^{-2/5}$~\cite{Ljung1999system}. 
In the setting of rational functions, strict stability is equivalent to exponential stability, which, in turn, implies the existence of all high-order derivatives of $G$. In this case, it would be suboptimal to employ the naive estimator~\eqref{eq:naive_estimator} without additional smoothing. We leave this for future work.

The $\mathcal{H}_{\infty}$ norm bound picks up the frequency which is the hardest to learn as it depends on the worst-case SNR. Assume that the input excites uniformly all frequencies, that is, $\sigma_{u,\ell}=\sigma_u$, for all $\ell\in[M]$. Then, the worst case SNR scales inversely with $\sup_{\ell}\sqrt{\opnorm{\Phi_{v,N}(\ell N_p)}}$; as we increase the frequency resolution $M$, this quantity scales, in turn, with the $\mathcal{H}_{\infty}$ norm of the noise filter $H(\qq)$.

\begin{figure}[t]
 	\centering
 	\includegraphics[width=0.9\columnwidth]{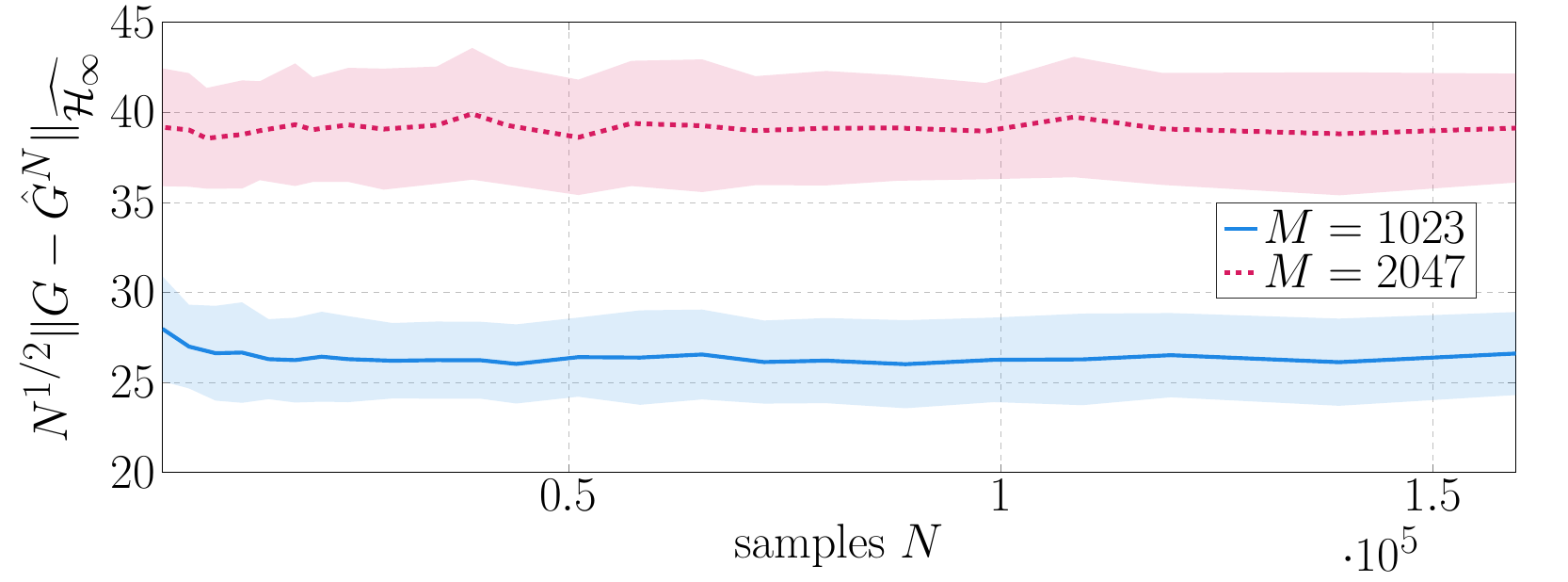}
 	\caption{The (normalized) empirical maximum error of the ETFE over the fixed frequency grid $2\pi \ell/M$, $\ell\in[M]$, for fixed $M$. The shaded areas show one (empirical) standard deviation. It decays with a rate of $N^{-1/2}$. Moreover, the error increases as we require more resolution, i.e., larger $M$. The error for $M=2047$ is roughly $\sqrt{2}$ times larger than for $M=1023$, verifying the result of Theorem~\ref{thm:finite_sample_ETFE}.}
 	\label{fig:discretized_points}
 \end{figure}

 \begin{figure}[t]
 	\centering
 	\includegraphics[width=0.9\columnwidth]{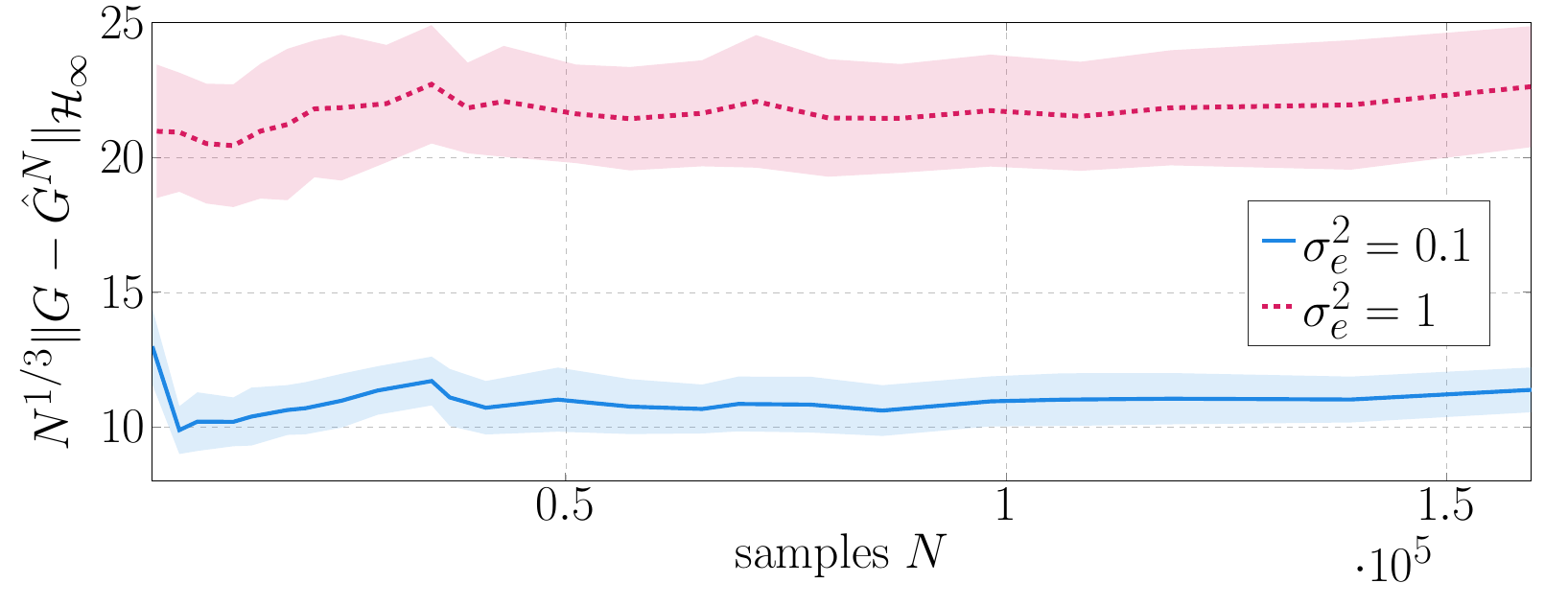}
 	\caption{The (normalized) empirical $\mathcal{H}_\infty$ norm of the ETFE, based on the naive estimator~\eqref{eq:naive_estimator}.  The shaded areas show one (empirical) standard deviation. We optimize the value of $M$ for every choice of $N$. Unlike the error at fixed frequencies, it decays slower, with a rate of $N^{-1/3}$.}
 	\label{fig:uniform_bounds}
 \end{figure}
\section{Simulations}\label{sec:simulation}
We study the performance of the ETFE and the naive estimator via a numerical example. Consider the system
\begin{equation*}
G(\qq)=\frac{0.12\qq^{-1}+0.18\qq^{-2}}{1-1.4\qq^{-1}+1.443\qq^{-2}-1.123\qq^{-3}+0.7729q^{-4}}
\end{equation*}
with noise filter $H(\qq)=(1-0.2\qq^{-1})^{-1}$. Let all past inputs be zero $u_t=0,t<0$. 
We generate the excitation signal based on PRBS~\cite{Ljung1999system} with an additional offset to excite the zero frequency. Note that PRBS  maximal length signals require $M=2^d-1$, for some $d\in\mathbb{N}$. 

In the first simulation, we study the maximum error of the ETFE over the fixed grid $\{2\pi \ell/M\}_{\ell\in[M]}$
\[
\snorm{G-\hat{G}^N}_{\widehat{\mathcal{H}_{\infty}}}\triangleq \sup_{\ell\in[M]}\opnorm{G(e^{j2\pi\frac{\ell}{M}})-\hat{G}_{\ell\Np}},
\]
which can be thought as the ``discretized" $\mathcal{H}_{\infty}$ norm of the error.  We keep $M$ fixed, $\sigmae^2=0.1$, and we vary the total number of samples $\Ntot=N$. To visualize the results, we perform $100$ Monte Carlo iterations for every number of samples $N$ and we present the empirical mean along with one empirical standard deviation. As shown in Fig.~\ref{fig:discretized_points}, the error decays with a rate of $N^{-1/2}$ validating Theorem~\ref{thm:finite_sample_ETFE}.

In the second simulation, we study the maximum error of the naive estimator~\eqref{eq:naive_estimator} over all frequencies, that is, the ``true" $\mathcal{H}_{\infty}$ norm of the error. For every number of samples $N$, we tune $d$ in $M=2^d-1$ to be the optimal one (empirically), and we perform $100$ Monte Carlo simulations. As shown in Fig.~\ref{fig:uniform_bounds}, the error decays with a rate of $N^{-1/3}$ reflecting the result of Theorem~\ref{thm:sup_norm_guarantees}.

\section{Conclusion and Future Work}
We provide finite-sample guarantees for the ETFE over a selected frequency grid, in the case of open-loop periodic excitation and under strict stability assumptions. By tuning the frequency resolution and exploiting Lipschitz continuity, we also obtain estimation guarantees in the $\mathcal{H}_{\infty}$ norm. 
An interesting direction for future work is studying finite-sample non-parametric least squares~\cite{wainwright2019high,ziemann2022single} in the frequency domain. This approach could lead to interesting connections between function class complexity and experiment design. Moreover, adding more structure, beyond Lipschitz continuity, will lead to faster rates. Other topics that are left for future work include the estimation of the noise statistics and extending the guarantees to different variations of the excitation method and the ETFE~\cite{falconi2023mean}. Finally, extending minimax lower bounds to this setting is also open~\cite{tu2019minimax}.
\section*{Acknowledgment}
This work has been supported by the Swiss National Science Foundation under NCCR Automation (grant agreement 51NF40 180545), and by the  European Research Council under the ERC Advanced grant agreement  787845 (OCAL).
\bibliographystyle{IEEEtran}
\bibliography{literature.bib}
\appendix

\section{Proof of Theorem~\ref{thm:HW}}
Recall that $\mathcal{L}_2(p)$ represents the Hilbert space of $p$-dimensional real-valued square summable sequences. For any sequences $\alpha,\beta\in\mathcal{L}_2(p)$, their inner product is defined as
\[
\langle \alpha,\beta\rangle\triangleq \sum_{k=1}^{\infty}\alpha^\top_k\beta_k,
\]
with the respective inner product norm
\[
\snorm{\alpha}=\sqrt{\langle \alpha,\alpha\rangle}.
\]
Let $\A:\mathcal{L}_2(p)\rightarrow \real^d$ be a linear map from sequences to finite-dimensional vectors. 
Consider the standard orthonormal basis for $\mathcal{L}_2(p)$, that is, the set $\{b_i\}_{i=1}^{\infty}$, where $b_i=(0,0,\dots,1,0,\dots)$ is the unit norm sequence with all elements $0$ except for the $i$-th one. Then, we can represent $\A$ with a semi-infinite block matrix
\begin{equation}\label{eq:semi_infinite_matrix_representation}
\begin{bmatrix}
\A_{11}&\A_{12}&\cdots\\
\A_{21}&\A_{22}&\cdots\\
\vdots\\
\A_{d1}&\A_{d2}&\cdots
\end{bmatrix},
\end{equation}
where $\A_{ij}\in\real^{1\times p}$, $i=1,\dots,d$, $j\ge 1$.
Then, the Frobenius norm is given by
\[
\hsnorm{\A}^2\triangleq \sum_{i=1}^d\sum_{j=1}^{\infty}\hsnorm{\A_{ij}}^2.
\]
In the following, we will study operators with bounded Frobenius norm $\hsnorm{\A}<\infty$; the operator norm is upper bounded by the Frobenius norm, and, thus, also bounded.

Consider now a sequence $z\triangleq \{z_t\}_{t=1}^{\infty}$ of independent zero mean $K-$sub-Gaussian random variables, with $z_t\in\real^p$. 
We will study quadratic forms $\snorm{\A(z)}^2$. Since $z$ is not square summable, $\A(z)$ is interpreted in terms of the extension
\begin{equation}\label{eq:extension_of_A}
[\A(z)]_i=\lim_{k\rightarrow \infty}\sum_{j=1}^k{\A_{ij}z_j},\,i=1,\dots,d
\end{equation}
when the above limit exists.  
We note that in our setting $\A(z)$ is well-defined.
\begin{lemma}[Well-posedness]\label{lem:well_posed}
Let $\A:\mathcal{L}_2(p)\rightarrow \real^d$ be a linear map with bounded Frobenius norm. Let $z=\{z_t\}_{t=1}^{\infty}$ be zero-mean, independent, and $K-$sub-Gaussian. Then, the extension $\A(z)\in\real^{d}$ is well-defined almost surely.
\end{lemma}
\begin{proof}
    Since $z_t$, $t\ge 1$ are $K^2-$sub-Gaussian, we have
    \[
    \sup_{t\ge 1}\E z^\top_tz_t=c p K^2<\infty,
    \]
    for some universal constant $c$. This follows from the fact that every coordinate $[z_{t}]_s$, $s=1,\dots,p$ is also $K^2-$sub-gaussian and the bound $\E ([z_{t}]_s)^2\le c K^2$ (see~\cite[Prop. 2.5.2]{vershynin2018high}). 
    Fix an $i$, $i=1,\dots,d$. We have
    \begin{equation*}
\sum_{j=1}^{k}\mathrm{var}(\A_{ij}z_j)
=\sum_{j=1}^{k}\E z_j^\top \A^\top_{ij}\A_{ij}z_j
=\sum_{j=1}^{k}\snorm{\A_{ij}}^2\E\snorm{z_j}^2\le \sum_{j=1}^{k}\snorm{\A_{ij}}^2 cp K^2
    \end{equation*}
    Hence, we have
    \[
\sum_{j=1}^{\infty}\mathrm{var}(\A_{ij}z_j)\le \hsnorm{\A}^2 cpK^2<\infty.
    \]
   Hence, by the Kolmogorov maximal inequality~\cite[Th. 2.5.3]{durrett2010probability}, the series
$
\lim_{k\rightarrow \infty}\sum_{j=1}^k{\A_{ij}z_j}
$
converges almost surely for any $i=1,\dots,d$ and, thus, $\A(z)$ is almost surely well-defined.
\end{proof}
The main idea behind proving Theorem~\ref{thm:HW} is using the Chernoff bound method. Let $X=\snorm{\A(z)}^2-\E\snorm{\A(z)}^2$. Then, for any $\lambda\ge0$
\[
\Pr(X\ge t)\le \min_{\lambda\ge 0}e^{-\lambda t}\E e^{\lambda X}.
\]
To control the right-hand side, we have to upper bound the moment-generating function $\E e^{\lambda X}$, which requires most of the work.

\begin{lemma}[MGF bound]\label{lem:mgf_HW}
Let $z=\{z_t\}_{t=1}^{\infty}$ be a sequence of independent, zero-mean, $K^2$-sub-Gaussian random variables. Let $\A:\mathcal{L}_2(p)\rightarrow \real^d$ be a linear map with bounded Frobenius norm $\hsnorm{\A}<\infty$. For every $\lambda \in\real$ such that $\abs{\lambda}\le (8\sqrt{2}K^2 \opnorm{\A}^2)^{-1}$, we have
\[
\E \exp(\lambda \snorm{\A(z)}^2-\lambda\E\snorm{\A(z)}^2)\le 36 \lambda^2 K^4 \opnorm{\A}^2\hsnorm{\A}^2,
\]
where the extension $\A(z)$ is interpreted as in~\eqref{eq:extension_of_A}.
\end{lemma}
\begin{proof}
We will consider finite-dimensional truncations of the quadratic form $\snorm{\A(z)}^2$, apply the result of~\cite{ziemann2023tutorial}, and take limits.
 Let $\hat{\A}^{(k)}:\mathcal{L}_2(p)\rightarrow \real^d$ be the truncated map with semi-infinite matrix representation
 \begin{equation}
 \hat{\A}^{(k)}_{ij}=\left.\begin{aligned}
&\A_{ij},\,\text{for }i=1,\dots,d,\,j\le k\\
&0\text{ elsewhere}
 \end{aligned}\right\}.
 \end{equation}
Define the random variables
\begin{align*}
X_k&\triangleq \exp(\lambda \snorm{\hat{\A}^{(k)}(z)}^2-\lambda\E\snorm{\hat{\A}^{(k)}(z)}^2),\,k\ge 1\\
X&\triangleq \exp(\lambda \snorm{\A(z)}^2-\lambda\E\snorm{\A(z)}^2).
\end{align*}
\textbf{Step a): truncation.} First, we will establish the desired inequality for all $X_k$. Notice that $\snorm{\hat{\A}^{(k)}(z)}^2$ is a finite dimensional quadratic form. Hence, Proposition A.1 of~\cite{ziemann2023tutorial} applies and provides the bound
\[
\E X_k\le 36 \lambda^2 K^4 \hsnorm{(\hat{\A}^{(k)})^* \hat{\A}^{(k)}}^2,
\]
for $\abs{\lambda}\le (8\sqrt{2}K^2 \opnorm{(\hat{\A}^{(k)})^* \hat{\A}^{(k)}})$ where $(\hat{\A}^{(k)})^*$ denotes the adjoint. Truncation does not increase the norm, i.e., $\hsnorm{\A}\ge \hsnorm{\hat{\A}^{(k)}}$ and $\opnorm{\A}\ge \opnorm{\hat{\A}^{(k)}}$. Moreover
\[
\hsnorm{(\hat{\A}^{(k)})^* \hat{\A}^{(k)}}^2\le \opnorm{\hat{\A}^{(k)}}^2\hsnorm{\hat{\A}^{(k)}}^2.
\]
Hence, we can write
\begin{equation}\label{eq:HW_truncated}
\E X_k\le 36 \lambda^2 K^4 \opnorm{\A}^2\hsnorm{\A}^2,
\end{equation}
for $\abs{\lambda}\le (8\sqrt{2}K^2 \opnorm{\A}^2)$.

\textbf{Step b): convergence.} We prove that $X_k$ converges almost surely to $X$. Convergence of $\hat{\A}^{(k)}(z)$ to $\A(z)$ follows from Lemma~\ref{lem:well_posed}. To show convergence of  $\E\snorm{\hat{\A}^{(k)}(z)}^2$ observe that
\begin{equation*}
|\snorm{\A(z)}^2-\snorm{\hat{\A}^{(k)}(z)}^2|\le 2\abs{\paren{\hat{\A}^{(k)}(z)}^\top(\A-\hat{\A}^{k})(z)}+\snorm{(\A-\hat{\A}^{k})(z)}^2.
\end{equation*}
Let $C_z=\sup_{t\ge 1}\E z^\top_t z_t$; due to sub-Gaussianity $C_z=cpK^2$ for some universal constant $c$--see proof of  Lemma~\ref{lem:well_posed}. Invoking the Cauchy-Schwarz inequality
\begin{equation*}
\E|\snorm{\A(z)}^2-\snorm{\hat{\A}^{(k)}(z)}^2|\le 2C_z \hsnorm{\A}\hsnorm{\A-\hat{A}^{(k)}}+C_z\hsnorm{\A-\hat{A}^{(k)}}^2\rightarrow 0,
\end{equation*}
as $k\rightarrow \infty$. Convergence of $X_k$ to $X$ follows by continuity.

\textbf{Step c): limit.} Since $X_k>0$, by Fatou's lemma
\begin{equation}\label{eq:Fatous}
\E X\le \liminf \E X_k.
\end{equation}
The result follows from~\eqref{eq:HW_truncated},~\eqref{eq:Fatous}.
\end{proof}

The proof of Theorem~\ref{thm:HW} now follows from Lemma~\ref{lem:mgf_HW} and the Chernoff bound method. The remaining details are omitted since they are identical to~\cite[proof of Th. 2.1 and Th. B.1]{ziemann2023tutorial}. $\qed$

\section{Proof of Lemma~\ref{lem:noise_concentration}}
\begin{proof}
First, we will establish a high probability bound on $V_kV^*_K$. Then we we use that to obtain a high probability bound on the operator norm $\opnorm{V_k}$.

\textbf{Step 1): Bound on $V_kV^*_k$.} Due to the independence between the experiments
\begin{equation}\label{eq:covariance_of_Vk}
\E V^*_kV_k=\diag(\E (V^{(i)})^*_kV^{(i)}_k)=\tr (\Phi_{v,N}(k)) I_{\du}.
\end{equation}
By the variational representation of the operator norm for symmetric matrices
\begin{align}
&\opnorm{V^*_kV_k-\tr (\Phi_{v,N}(k)) I_{\du}}\nonumber\\
&=\sup_{\xi \in \cplx^{\du},\snorm{\xi}=1}\abs{\xi^*(V^*_kV_k-\tr (\Phi_{v,N}(k)) I_{\du})\xi}\label{eq:variational_representation_symmetric_matrix}\\
&=\sup_{\xi \in \cplx^{\du},\snorm{\xi}=1}\abs{\snorm{V_k\xi}^2-\E\snorm{V_k\xi}^2}\nonumber.
\end{align}
We need to control the quantity inside the supremum over the whole unit sphere. First, we will bound it for a fixed $\xi$ using the Hanson-Wright inequality. Then, we will discretize the unit sphere and use a covering argument.

\textbf{Step 1a): fixed $\xi$.} Observe that $V^{(i)}_k=\mathcal{F}_k((h*e^{(i)}))$ is the output of a linear map applied to the infinite sequence $e^{(i)}\triangleq \{e^{(i)}_{N-1},e^{(i)}_{N-2},\dots\}$, where $*$ denotes convolution. We can rewrite the above linear relation using the notation $V^{(i)}_k=H_k (e^{(i)})$, where $H_k$ represents the map $\mathcal{F}_k\circ((h*\cdot))$.
For convenience, let us group all noise sequences $e^{(i)}$ into one single sequence $e$. We use the enumeration
\begin{align*}
e&\triangleq \set{e_t}_{t=1}^{\infty}\triangleq \set{e^{(1)}_{N-1},\dots,e^{(\du)}_{N-1},e^{(1)}_{N-2},\dots,e^{(\du)}_{N-2},\dots}.
\end{align*}
Let $\xi\in\cplx^{\du}$ be any vector with unit norm $\snorm{\xi}=1$.  
Then,
\[
V_k\xi=H_k (\sum_{i=1}^{\du} e^{(i)}\xi_i)=(H_k\circ \Xi)(e),
\]
where $\Xi$ is the linear map from $e$ to $\sum_{i=1}^{\du}e^{(i)}\xi_i$. We can represent $\Xi$ with the infinite matrix
\[
\Xi=\begin{bmatrix} \xi^* \otimes I_{\de}&0&0&\cdots\\0&\xi^* \otimes I_{\de}&0&\cdots\\0&0&\xi^* \otimes I_{\de}\\\vdots&&&\ddots
\end{bmatrix},
\]
where $\otimes$ denotes the Kronecker product.
Since $\xi^*\xi=1$, $\Xi\circ\Xi^*$ is the identity operator and has unit norm. Note that $(H_k\circ \Xi)(e)\in\cplx^{\dy}$ is complex-valued, while the Hanson-Wright inequality applies to real valued maps. To resolve this issue, we embed $V_k\xi$ to a real vector space. Let $\Re (H_k\circ \Xi)$ denote the real part of the operator $H_k\circ \Xi$ with $\Im (H_k\circ \Xi)$ the respective imaginary part. Define the lifted operator
\[
\mathcal{A}(e)\triangleq \begin{bmatrix}
    \Re (H_k\circ \Xi)(e)\\
    \Im (H_k\circ \Xi)(e)
\end{bmatrix}\in\real^{2\dy}.
\]
Our goal is apply the Hanson-Wright inequality to $\mathcal{A}(e)$.
Observe that we have
\begin{equation}\label{eq:lifted_space_norm}
\snorm{\mathcal{A}(e)}^2=\snorm{ \Re (H_k\circ \Xi)(e)}^2+\snorm{ \Im (H_k\circ \Xi)(e)}^2=\xi^* V^*_kV_k\xi=\snorm{V_k\xi}^2.
\end{equation}
Using the representation of $\mathcal{A}$ as a semi-infinite matrix as in~\eqref{eq:semi_infinite_matrix_representation} and since the noise sequence $e=\{e_t\}_{t=1}^{\infty}$ is i.i.d. with covariance $\sigmae^2 I_{\de}$, we have
 \[
\E \snorm{\mathcal{A}(e)}^2=\E \sum_{i=1}^{2\dy}(\sum_{j=1}^{\infty}\mathcal{A}_{ij}e_j)^2=\sigmae^2\sum_{i=1}^{2\dy}\sum_{j=1}^{\infty}\hsnorm{\mathcal{A}_{ij}}^2 =\sigmae^2\hsnorm{\mathcal{A}}^2.
 \]
Meanwhile, we have
\[\E \xi^\top V^*_kV_k\xi\stackrel{\eqref{eq:covariance_of_Vk}}{=}\xi^\top\tr(\Phi_{v,N}(k))I_{\du}\xi=\tr(\Phi_{v,N}(k))\] 
Combining the above equalities results in
 \begin{equation}\label{eq:forbenius_mathcal_A}
\E \snorm{\mathcal{A}(e)}^2=\tr(\Phi_{v,N}(k)),\,\hsnorm{\mathcal{A}}^2=\tr(\Phi_{v,N}(k))\sigmae^{-2}.
 \end{equation}
Next, we upper-bound the operator norm of $\mathcal{A}$. Since the sequence $e$ is real valued, we have
 \[
\opnorm{\mathcal{A}}= \opnorm{H_k\circ \Xi}\le \opnorm{H_k},
 \]
 where the inequality follows from the fact that $\Xi$ has unit norm. Finally, since $V^{(i)}_k=H_ke^{(i)}$, for any $i=1,\dots,\du$, we have
 \[
\Phi_{v,N}(k)=\E V^{(i)}_k(V^{(i)})^*_k=\sigmae^2 H_kH^*_k.
 \]
 Combing the above we obtain
  \begin{equation}\label{eq:opnorm_mathcal_A}
\opnorm{\mathcal{A}}\le \sqrt{\opnorm{\Phi_{v,N}(k)
}}/\sigmae.
 \end{equation}

 With~\eqref{eq:lifted_space_norm},~\eqref{eq:forbenius_mathcal_A},~\eqref{eq:opnorm_mathcal_A} at hand we can finally apply the Hanson-Wright inequality  to $\mathcal{A}(e)$. Invoking Theorem~\ref{thm:HW} with $\alpha=\beta\sigmae$, we obtain
\begin{multline}
    \Pr\set{\big|\snorm{V_k\xi}^2-\tr(\Phi_{v,N}(k))\big|>\beta \: \tr(\Phi_{v,N}(k))}\le 2    \exp\paren{-144^{-1}\min\{\beta^2,\beta\}\frac{\sigmae^4 }{K^4}\frac{\tr(\Phi_{v,N}(k))}{\opnorm{\Phi_{v,N}(k)}}},
\end{multline}
Note that we simplified the expression by using $1/144\le 1/(16\sqrt{2})$ and by exploiting the fact that the sub-Gaussian parameter upper-bounds the variance, that is, $K\ge \sigmae$.

\textbf{Step 1b): covering argument.} To control the supremum in~\eqref{eq:variational_representation_symmetric_matrix}, we need to discretize the unit sphere. In particular, we consider an $\epsilon-$net and apply the following result.
\begin{lemma}[Operator Norm on a net~\cite{vershynin2018high,ziemann2023tutorial}]\label{lem:operator_norm_covering}
    Let $A\in\cplx^{d\times d}$ be a hermitian random matrix and let $\epsilon\in(0,1/2)$. Let $\mathcal{N}$ be an $\epsilon-$net of the complex unit sphere with minimal cardinality. Then, for any $t>0$
    \[\Pr(\opnorm{A}\ge t)\le \abs{\mathcal{N}} \max_{z\in\mathcal{N}}\Pr(\abs{z^*Az}\ge (1-2\epsilon)t),\]
    where the cardinality $\abs{\mathcal{N}}$ is upper bounded
    \[
     \abs{\mathcal{N}}\le \paren{1+\frac{2}{\epsilon}}^{2d}.
    \]
\end{lemma}
Lemma~\ref{lem:operator_norm_covering} follows directly by Lemma 2.5 in~\cite{ziemann2023tutorial}. The exponent scales with $2d$ instead of $d$ since the complex unit sphere in $\cplx^d$ is isomorphic to the unit sphere $S^{2d-1}$ in $\real^{2d}$.

Applying Lemma~\ref{lem:operator_norm_covering} with $\epsilon=1/4$, $t$ replaced with $2\beta$, and $A$ replaced with $V^*_kV_k-\tr(\Phi_{v,N}(k))I_{\du}$, we obtain
    \begin{multline}\label{eq:noise_concentration_squared}
    \Pr\set{\opnorm{V^*_kV_k-\tr(\Phi_{v,N}(k))I_{\du}}>2\beta\:  \tr(\Phi_{v,N}(k))}\le \\9^{2\du} 2    \exp\paren{-144^{-1}\min\left\{\beta^2, \beta \right\}\frac{\sigmae^4 }{K^4}\frac{\tr(\Phi_{v,N}(k))}{\opnorm{\Phi_{v,N}(k)}}}.
\end{multline}
\textbf{Step 2): Bound on the operator norm $\opnorm{V_k}$.} We invoke the following lemma.
\begin{lemma}[(3.2) in~\cite{vershynin2018high}]\label{lem:elementary_inequality}
Let $z,s\ge 0$ be any two positive real numbers. If $\abs{z-1}\ge s$ then $\abs{z^2-1}\ge \max\{s,s^2\}$.
\end{lemma}
Let $s^2=\min\{\beta,\beta^2\},s\ge 0$, which is equivalent to $\beta=\max\{s,s^2\}$.
Then, the event 
\begin{equation*}
\mathcal{E}_1\triangleq\set{\Big|\opnorm{V_k}/\sqrt{\tr(\Phi_{v,N}(k))}-1\Big|\ge 2 s}\end{equation*}
implies that there exists a unit norm $\xi\in\real^{\du}$ such that
\[\Big|\snorm{V_k\xi}/\sqrt{\tr(\Phi_{v,N}(k))}-1\Big|\ge 2 s.\]
By Lemma~\ref{lem:elementary_inequality}, this implies that
\[
\Big|\snorm{V_k\xi}^2/\tr(\Phi_{v,N}(k))-1\Big|\!\ge\! \max\{2 s,4s^2\}\ge 2\max\{s,s^2\}.
\]
This together with~\eqref{eq:variational_representation_symmetric_matrix} imply the following event holds
\begin{equation*}
\mathcal{E}_2\triangleq\set{\opnorm{V^*_kV_k-\tr(\Phi_{v,N}(k))I_{\du}}>2\beta\:  \tr(\Phi_{v,N}(k))}.\end{equation*}
Since $\mathcal{E}_1\subset\mathcal{E}_2$ the result follows from~\eqref{eq:noise_concentration_squared}, where we replace $\min\{\beta^2,\beta\}$ with $s^2$.
\end{proof}
\section{Proof of Lemma~\ref{lem:PE}}
Fix an experiment index $i$. Since $u^{(i)}_t$ is periodic with period $M$ and $\Np$ is an integer, we have
\begin{equation*}
\begin{aligned}
   \sum_{t=0}^{N-1}u^{(i)}_te^{-j\omega_k t}&=\sum_{\tau=0}^{\Np-1}\sum_{t=0}^{M-1}u^{(i)}_te^{-j\omega_k (M\tau+t)}\\
   &=\paren{\sum_{\tau=0}^{\Np-1}e^{-j\frac{2\pi k}{\Np}\tau}}\paren{\sum_{t=0}^{M-1}u^{(i)}_te^{-j\omega_k t}}
\end{aligned}
\end{equation*}
Therefore, we obtain
\begin{equation*}
\begin{aligned}
   U^{(i)}_{\Np\ell}&=\sqrt{\Np}\tilde{U}^{(i)}_l,\,\text{for }\Np\ell, \ell\in[M]\\
   U^{(i)}_k&=0,\,\text{elsewhere. }
\end{aligned}
\end{equation*}
For $k=\Np\ell$ we have
\[
U_kU^*_k=\Np \tilde{U}_\ell\tilde{U}^*_\ell\succeq N_p \sigma^2_{u,\ell}I_{\du}.\tag*{\qed}
\]

\section{Proof of Theorem~\ref{thm:finite_sample_ETFE}}
By~\eqref{eq:estimation_error}, the estimation error is upper bounded by
\[
\opnorm{\hat{G}_k-G(e^{j\omega_k})}\le \opnorm{V_k}\opnorm{U_k^{-1}}+\opnorm{T_{k,N}}\opnorm{U^{-1}_k}.
\]
Invoking Lemma~\ref{lem:noise_concentration}, we have
\[\Pr\paren{\opnorm{V_k}\ge(1+2s)\sqrt{\tr(\Phi_{v,N}(k))}}\le 9^{2\du}2\exp\paren{-\frac{1}{144}s^2\frac{\sigmae^4 }{K^4}\frac{\tr(\Phi_{v,N}(k))}{\opnorm{\Phi_{v,N}(k)}}}\]
Now fix a failure probability $\tilde{\delta}$. Choosing $s$ such that
\[
\frac{1}{144}s^2\frac{\sigmae^4}{K^4}\frac{\tr(\Phi_{v,N}(k))}{\opnorm{\Phi_{v,N}(k)}}\ge 4\du\log 9+\log 1/\tilde{\delta}
\]
guarantees that the above probability is upper bounded by $\tilde{\delta}$. Hence, with probability at least $1-\tilde{\delta}$
\[
\opnorm{V_k}\le \sqrt{\tr(\Phi_{v,N}(k))}+\sqrt{144}\frac{K^2}{\sigmae^2}\sqrt{\opnorm{\Phi_{v,N}(k)}}\sqrt{4\log 9\,\du+\log 1/\tilde{\delta}}.
\]
To simplify the final expression define the universal constant
\[
c=\sqrt{144}\sqrt{4\log 9}.
\]
Then, we obtain that with probability at least $1-\tilde{\delta}$
\begin{equation}\label{eq:pre_union_bound}\opnorm{V_k}\le \sqrt{\tr(\Phi_{v,N}(k))}+c\frac{K^2}{\sigmae^2}\sqrt{\opnorm{\Phi_{v,N}(k)}}\sqrt{\du+\log 1/\tilde{\delta}}.\end{equation}

To finish the proof we need to establish the upper bound~\eqref{eq:pre_union_bound} over all $k=\Np\ell,$ $\ell\in[M]$. 
It is sufficient to set $\tilde{\delta}=\delta/M$ and consider a union bound over all $\ell\in[M]$. Then, with probability at least $1-\delta$, for all $k=\Np\ell$, $\ell\in[M]$ 
\begin{equation*}
\opnorm{V_k}\le \sqrt{\dy}\sqrt{\opnorm{\Phi_{v,N}(k)}}+c\frac{K^2}{\sigmae^2}\sqrt{\opnorm{\Phi_{v,N}(k))}}\sqrt{\log(M/\delta)+\du},
\end{equation*}
where we used $\tr(\Phi_{v,N}(k))\le \dy\opnorm{\Phi_{v,N}(k)}$.
The proof now follows from Lemma~\ref{lem:deterministic_transient} and Lemma~\ref{lem:PE}.\hfill$\qed$
\section{Proof of Theorem~\ref{thm:sup_norm_guarantees}}
It follows directly from Theorem~\ref{thm:finite_sample_ETFE} and
\[
\opnorm{G(e^{j\omega})-\hat{G}^N(e^{j\omega})}\le \frac{\pi}{M}\snorm{G}_{\star}+\opnorm{G(e^{j\omega_{\ell N_p}})-\hat{G}_{\ell N_p}}.\tag*{\qed}
\]

\end{document}